\documentclass{article}
\pdfoutput=1

\usepackage[nonatbib,preprint]{neurips_2020}

\usepackage[utf8]{inputenc} 
\usepackage[T1]{fontenc}    
\usepackage{hyperref}       
\usepackage{url}            
\usepackage{booktabs}       
\usepackage{amsfonts}       
\usepackage{nicefrac}       
\usepackage{microtype}      
\usepackage{xspace}
\usepackage{hyperref}

\hypersetup{
    citecolor=blue,
    colorlinks=true,
    linkcolor=blue,
    filecolor=blue, urlcolor=blue,
}
\usepackage{float}

\usepackage{multirow}
\usepackage[colorinlistoftodos]{todonotes}

\usepackage{subfig}
\usepackage{algorithm}
\usepackage[noend]{algpseudocode}

\usepackage{soul}
\usepackage[style=base]{caption}
\usepackage[toc,page]{appendix}
\usepackage{mathtools}

\usepackage{amssymb}

\newtheorem{theorem}{Theorem}

\newtheorem{corollary}[theorem]{Corollary}

\newtheorem{definition}[theorem]{Definition}

\newtheorem{lemma}[theorem]{Lemma}

\newtheorem{proposition}[theorem]{Proposition}
\newtheorem{remark}[theorem]{Remark}

\newenvironment{proof}[1][Proof]{\noindent\textbf{#1.} }{\ \rule{0.5em}{0.5em}}
\usepackage{hyperref}
\hypersetup{
    colorlinks,
    citecolor=black,
    filecolor=black,
    linkcolor=black,
    urlcolor=black
}

\usepackage[normalem]{ulem}

\newcommand{\newrj}[1]{{\color{blue} #1}}

\newcommand{\Var}{\ensuremath{\operatorname{Var}}}
\newcommand{\OnlineETI}{{\ensuremath{\mathsf{OnlineETI}}}\xspace}

\title{Adaptive Experimental Design with Temporal\\ Interference: A Maximum Likelihood Approach}

\author{%
  Peter Glynn, Ramesh Johari, Mohammad Rasouli\\
   Department of Management Science and Engineering,\\
  Stanford University, Stanford, CA 94305 \\
  \texttt{glynn, rjohari, rasoulim  @stanford.edu} \\
}

\begin{document}

\maketitle

\begin{abstract}
Suppose an online platform wants to compare a treatment and control policy, e.g., two different matching algorithms in a ridesharing system, or two different inventory management algorithms in an online retail site.  Standard randomized controlled trials are typically not feasible, since the goal is to estimate policy performance on the entire system.  Instead, the typical current practice involves dynamically alternating between the two policies for fixed lengths of time, and comparing the average performance of each over the intervals in which they were run as an estimate of the treatment effect.  However, this approach suffers from {\em temporal interference}: one algorithm alters the state of the system as seen by the second algorithm, biasing estimates of the treatment effect.  Further, the simple non-adaptive nature of such designs implies they are not sample efficient.

We develop a benchmark theoretical model in which to study optimal experimental design for this setting.  We view testing the two policies as the problem of estimating the steady state difference in reward between two unknown Markov chains (i.e., policies).  We assume estimation of the steady state reward for each chain proceeds via nonparametric maximum likelihood, and search for consistent (i.e., asymptotically unbiased) experimental designs that are efficient (i.e., asymptotically minimum variance).  Characterizing such designs is equivalent to a Markov decision problem with a minimum variance objective; such problems generally do not admit tractable solutions.  Remarkably, in our setting, using a novel application of classical martingale analysis of Markov chains via Poisson's equation, we characterize efficient designs via a succinct convex optimization problem.  We use this characterization to propose a consistent, efficient online experimental design that adaptively samples the two Markov chains.

\end{abstract}

\section{Introduction}
\label{sec:intro}

Suppose an online platform wants to compare a treatment and control policy, e.g., two different matching algorithms in a ridesharing system, or two different inventory management algorithms in an online retail site.  Standard randomized controlled trials are typically not feasible, since the goal is to estimate policy performance on the entire system.  Instead, the typical current practice involves dynamically alternating between the two policies for fixed lengths of time, and comparing the average performance of each over the intervals in which they were run as an estimate of the treatment effect; this is referred to as a {\em switchback} experimental design  \cite{brandt1938tests, lucas1956switchback}.

However, switchback designs suffer from {\em temporal interference}: the {\em initial condition} in each interval of each policy is determined by the {\em previous} interval of the other policy, and so standard estimation techniques in this setting are {\em biased}. Bias due to temporal interference has been observed in ride-sharing \cite{Lyft}, in delivery services \cite{Doordash}, and ad auctions \cite{basse2016randomization}.  Further, the simple, non-adaptive nature of such designs implies they are not {\em sample efficient}.  Optimal consistent, efficient experimental design in this setting has remained a significant theoretical and practical challenge.

Our paper provides an optimal experimental design within a benchmark theoretical model for settings with temporal interference (Section \ref{sec:model}).  The central challenge posed by temporal interference is the following: we are effectively allowed only one real-world run of the system, with only finitely many observations.  On the other hand, we need to use this single run to estimate performance of {\em both} the systems induced by each of the treatment and control policies.  We model the problem by viewing each policy as its own {\em Markov chain} on a common underlying state space.  The experimental design problem is then to estimate the difference in the steady state reward under the treatment and control Markov chains, using only one run of the system, and without prior knowledge of any of the parameters of either policy or their rewards.

Our key contribution is a  characterization of consistent and asymptotically efficient policies when estimation proceeds via maximum likelihood.  In particular, we model the unknowns nonparametricaly, and use an associated nonparametric maximum likelihood estimator (MLE; Section \ref{sec:mle}).  At any time step, given the current system state, the experimental design chooses which chain to sample.  We restrict attention to policies that satisfy a weak regularity requirement we call {\em time-average regularity} (TAR; Section \ref{sec:tar}).  n particular, with TAR policies, we show the MLE is consistent.

In Section \ref{sec:var}, we present our main result: we characterize efficient TAR policies, i.e., those for which the MLE achieves asymptotically minimum variance among all TAR policies.  
Our approach uses a novel application of classical martingale analysis of Markov chains via Poisson's equation, and leads to a characterization of efficient designs via a succinct convex optimization problem.  This simple characterization is somewhat remarkable: Markov decision problems (MDPs) with variance minimization as the objective have historically not admitted structurally simple solutions (see below).
In Section \ref{sec:online}, we use this characterization to construct an efficient, consistent adaptive online experimental design when estimation proceeds via the MLE.  We conclude in Section \ref{sec:conclusion}.

{\bf Related work}. {\em Interference} occurs in experiments whenever the outcome of a given experimental unit depends on the assignment status of {\em other} experimental units to either treatment or control.  Recent work has devoted extensive attention to interference in experimental design for {\em networks} \cite{athey2018exact, eckles2017design, hudgens2008toward, manski2013identification, sobel2006randomized} and {\em marketplaces} \cite{kohavi2009controlled, ostrovsky2011reserve, wager2019experimenting, johari2020twosided}. 

As our approach involves solving a MDP with unknown primitives, it has some model similarities with reinforcement learning, and particularly work in pure exploration in reinforcement learning \cite{putta2017pure, dann2017unifying}.   The main distinction in our work is that the objective for the MDP we solve is minimum variance of the MLE.

In general, the literature on MDPs with variance minimization as the objective demonstrates the principal of optimality and dynamic programming cannot be used in the classical form for average reward MDPs \cite{sobel1982variance, sobel1994mean, di2012policy, filar1989variance,iancu2015tight}.  In particular, the optimal policy is not necessarily Markov; it can be random; and finding the optimal policy is NP-hard \cite{mannor2011mean, yu2018approximate}.  In contrast to these prior results, our work has thus identified a remarkably tractable MDP with variance minimization as the objective.




\section{Preliminaries} 
\label{sec:model}

In this section we introduce the basic formal framework we employ throughout the paper.  We develop the relevant notation to describe two distinct Markov chains on a common state space, as well as the design of adaptive experiments to compare the long-run average reward of these two chains.  

{\bf Notation.} As is common in analysis of finite Markov chains, we view distributions as row vectors and reward vectors as column vectors as appropriate.  In addition, we use ``$\xrightarrow{p}$" to denote convergence in probability, and ``$\Rightarrow$'' to denote weak convergence of random variables.

{\bf Time.}  We assume time is discrete, and indexed by $n = 0, 1, 2, \ldots$.

{\bf State space}.  We assume a finite state space $S$.

{\bf Two Markov chains}.  We wish to compare two different Markov chains indexed by $\ell = 1,2$ evolving on this common state space, defined by transition matrices
\begin{align}
    P(\ell)=(P(\ell,x,y): x,y\in S), \quad \ell=1,2.
\end{align}
We assume both $P(1)$ and $P(2)$ are {\em irreducible}.  

{\bf Auxiliary randomness}.  We require two sources of randomness beyond the Markov chains themselves: one that is used to generate random rewards, and the other that is used to allow experimental designs to be randomized.  Accordingly, we presume the existence of mutually independent sequences of i.i.d.~uniform$[0,1]$ random variables $U_0, U_1, \ldots,$ and $V_0, V_1, \ldots$.

{\bf Sample space and filtration}.  The sample space is $\Omega=(S\times [0,1]^2)^\infty$, with $\omega \in \Omega$ written as $\omega = ( (x_n, u_n, v_n), n \geq 0)$.  For $n\ge 0$, set $X_n(\omega)=x_n$ and $U_n(\omega)=u_n$, $V_n(\omega)=v_n$.  We define the filtration $\mathcal{G}_n=\sigma\big((X_j, U_j, V_j): 0\le j\leq n\big)$ for $n\ge 0$.  We also let $\mathcal{G}_\infty = \sigma\big((X_j, U_j, V_j): j \geq 0\big)$.

{\bf Policies (experiment designs)}.  A sequence of random variables $A=(A_n: n\ge 0)$ is said to be a {\em policy} if $A_n\in \{1,2\}$ for $n\ge 0$, and $A$ is adapted to $(\mathcal{G}_n:n\ge 0)$.  A policy is also an {\em experimental design}: it determines how the experimenter chooses which chain to run at each time step.\footnote{In general, the experimenter may commit in advance to a time horizon $N$ of interest, and the experimenter may use this knowledge in design of their policy.  In what we study here, for ease of presentation, we presume that the policy is defined for all $n \geq 0$.  In the fixed horizon setting, our results can be extended in a straightforward manner to characterize consistency and variance as $N \to \infty$ under appropriate regularity conditions.}

Note that every policy induces a probability measure on $(\Omega, \mathcal{G}_\infty)$; this probability measure has conditional distributions defined as follows, for Borel subsets $A,B \subset [0,1]$ and with Lebesgue measure denoted $\mu$:
\[ P(X_{n+1} = x, U_{n+1} \in A, V_{n+1} \in B | \mathcal{G}_n) = P(A_n, X_n, y) \mu(A) \mu(B). \]

{\bf Rewards}.  When chain $\ell$ is in state $x$ and transitions to state $y$, a random reward is obtained, independent of the past.  Formally, denote the cumulative distribution function of the reward by  $F(\cdot|\ell, x, y)$. Then the reward at time $n\ge 0$ is:
\begin{align}
    R_n=F^{-1}(V_n | A_{n-1}, X_{n-1}, X_n).
\end{align}

For technical simplicity, we assume that the support of $F$ is {\em bounded}, i.e., that rewards are bounded in magnitude.

{\bf Stationary distributions}.  Because $S$ is finite and the two matrices $P(1)$ and $P(2)$ are irreducible, there exist unique stationary distributions $\pi(\ell)=(\pi(\ell,x): x\in S), \ell=1,2$ satisfying 
\begin{gather}
    \pi(\ell) =\pi(\ell)P(\ell);\\
    \pi(\ell,x)\ge 0, x\in S;\\
    \sum_{x\in S} \pi(\ell,x)=1.
\end{gather}

{\bf Long-run average reward}.  For $\ell = 1,2$, $x \in S$, define:
\begin{align}
    r(\ell,x)=\mathbb{E}\big\{R_{n+1}|X_n=x, A_n=\ell\big\}=\sum_{y\in S} P(\ell,x,y) \int_{\mathbb{R}} z F(dz|\ell,x,y).
\end{align}
Now define:
\[ \alpha(\ell)=\sum_{x \in S} \pi(\ell,x)r(\ell,x), \ell = 1,2. \]
This is the stationary average reward of chain $\ell$.  By the ergodic theorem for Markov chains, this is also the long-run average reward associated to chain $\ell$.  

{\bf Treatment effect}.  We are interested in the difference in long run average rewards between the two chains, i.e., $\alpha = \alpha(2) - \alpha(1)$.  This is the {\em treatment effect}.

{\bf Estimators}.  An {\em estimator} is a sequence of real-valued random variables $\hat{\alpha}=(\hat{\alpha}_n: n\ge 0)$ that is adapted to $(\mathcal{G}_n:n\ge 0)$.

Our goal is to design a combination of a policy $A$ and an estimator $\hat{\alpha}$ to estimate $\alpha=\alpha(2)-\alpha(1)$ {\em consistently} and {\em efficiently}, in senses that we make precise in the subsequent development.

\section{Maximum Likelihood Estimation}
\label{sec:mle}

In this section, we develop an approach to experiment design and estimation based on a maximum likelihood approach.  Given a policy, we develop the maximum likelihood estimator (MLE) for the treatment effect $\alpha$.  In particular, we take a nonparametric approach in this paper, as we make no parametric assumptions on the Markov chains being studied.  Thus our approach involves maximum likelihood estimation of the transition matrices, followed by inversion to obtain an MLE for the steady state distribution.



Let $\Gamma_n(\ell,x)$ to be the number of times action $i$ at state $x$ is sampled by time $n$:
\begin{align}
    \Gamma_n(\ell,x):=\sum_{j=0}^{n-1} I(X_j=x, A_j=\ell),\ x\in S,\ \ell=1,2.
\end{align}



Now define:
\begin{align}
\hat{P}_n(\ell,x,y)&=\frac{\sum_{j=0}^{n-1} I(X_j=x, A_j=\ell, X_{j+1}=y)}{\max\{\Gamma_n(\ell,x),1\}};\\
\hat{P}_n(\ell)&=(\hat{P}_n(\ell,x,y):x,y\in S).
\end{align}

The estimators $\hat{P}(1)$ and $\hat{P}(2)$ are standard maximum likelihood estimators (MLE) for the corresponding transition matrices $P(1)$ and $P(2)$.

Define the stopping time $J=\min\{n\ge 0: \hat{P}_n(\ell)$ is irreducible for $\ell=1,2\}$. Note that $\hat{P}_n(\ell)$ will remain irreducible for $n\ge J$, since any path with positive probability under $\hat{P}_J$ will have positive probability under $\hat{P}_n$ for all $n\ge J$.  Thus for each $n\ge J$, $\hat{P}_n(\ell)$ has a unique stationary distribution $\hat{\pi}_n(\ell)$ satisfying 
\begin{gather}
    \hat{\pi}_n(\ell)=\hat{\pi}_n(\ell)\hat{P}_n(\ell);\\
    \hat{\pi}_n(\ell,x)\ge 0, x\in S;\\
    \sum_{x \in S}\hat{\pi}_n(\ell,x)=1.
\end{gather}
Note that by equivariance of the MLE, since stationary distributions are functionals of the transition matrices, each $\hat{\pi}_n(\ell)$ is also the MLE for $\pi(\ell)$.

Define
\begin{align}
\hat{r}_n(\ell,x)&=\frac{\sum_{j=0}^{n-1} I(X_j=x, A_j=\ell)R_{j+1}}{\max\{\Gamma_n(\ell,x), 1\}}. 
\end{align}
The preceding is the MLE of $r(\ell,x)$ along the realized sample path.

Finally, for $n\ge J$, again by equivariance of the MLE, we conclude that the resulting nonparametric MLE $\hat{\alpha}_n$ for $\alpha$ is given by the following:
\begin{align}
    \hat{\alpha}_n=\hat{\pi}_n(2)\hat{r}_n(2)-\hat{\pi}_n(1)\hat{r}_n(1).
\end{align}
For $n< J$, $\ell = 1,2$, and $x \in S$, we arbitrarily define $\pi_n(\ell,x)=1 / |S|$, $\hat{r}_n(\ell,x)=0$, and $\hat{\alpha}_n = 0$.

\section{Time Average Regular (TAR) Policies}
\label{sec:tar}

We specialize our study to the following class of policies, that satisfy a mild regularity condition.  As noted at the end of this subsection in Corollary \ref{cor:consistency}, all TAR policies make $\hat{\alpha}_n$ a consistent estimator of $\alpha$.

\begin{definition}
\label{def:tar}
Policy $A$ is \emph{time-average regular (TAR)} with (possibly random) {\em policy limits} $\gamma=(\gamma(\ell,x): x\in S, \ell=1,2)$ if: 
\begin{align}
    \frac{1}{n} \Gamma_n(\ell,x)\xrightarrow{p} \gamma(\ell,x)
\end{align}
as $n\rightarrow \infty$ for each $x\in S, \ell=1,2$.
\end{definition}

In the sequel we typically require that $\gamma(\ell,x) > 0$ almost surely.  Note that in Definition \ref{def:tar}, in general the policy limits {\em will be dependent on the initial state $X_0$}.  We suppress this dependence in the notation, because this dependence on initial conditions will not play a significant role.  In particular, the policies we suggest for efficient experimentation will lead to deterministic policy limits, with no dependence on the initial state.

We now characterize the structure of policy limits; in particular, we show in Proposition \ref{prop: policy limits} below that policy limits almost surely lie in the set $\mathcal{K}$ defined next.  The proof is in the Appendix.

\begin{definition}\label{def: mathcal K}
Define the set $\mathcal{K}$ as follows:
\begin{align}
    \mathcal{K} & = \Big\{ \kappa = (\kappa(\ell, x) : x \in S, \ell = 1,2)\text{ such that: } \\
& \kappa(1,y)+\kappa(2,y)=\sum_{\ell=1}^{2}\sum_{x\in S} \kappa(\ell,x)P(\ell,x,y), \quad y\in S;\label{eq: kappa eq 1}\\
&    \sum_{\ell=1}^{2}\sum_{x\in S}\kappa(\ell,x)=1;\label{eq: kappa eq 2}\\
&    \kappa(\ell,x)\ge 0, \quad x\in S, \ell=1,2\Big\}.\label{eq: kappa eq 3}
\end{align}
\end{definition}

\begin{proposition}\label{prop: policy limits} 
Let $A$ be a time average regular policy with policy limits $\gamma=(\gamma(\ell,x): x\in S, \ell=1,2)$.  Then almost surely, $\gamma \in \mathcal{K}$.
\end{proposition}

Although straightforward, the preceding proposition encodes a surprising benefit of estimation using both chains.  In particular, experimental designs in this setting can benefit from {\em cooperative exploration}: one chain can be used to drive the system into states for which we want samples for the {\em other} chain (cf.~the relation \eqref{eq: kappa eq 1}).  In Appendix \ref{sec:example_cooperative}, we illustrate that this possibility can yield substantial benefits in estimation variance.  Indeed, as shown in Appendix \ref{sec:example_cooperative}, examples can be constructed such that the variance of the MLE of the treatment effect after $n$ time steps is unboundedly lower for the optimal policy, relative to the variance of the difference in the MLEs of steady state rewards obtained by running each chain in isolation for $n$ time steps.

\begin{remark}
\label{rem:markov_TAR}
Given a TAR policy $A$, for any initial state, the law of the resulting policy limits $\gamma$ is a probability measure over $\mathcal{K}$, according to Proposition \ref{prop: policy limits}.

Conversely, suppose that $\kappa\in \mathcal{K}$ is  positive (i.e. $\kappa(\ell,x)>0$ for $x\in S, \ell=1,2$). We show that regardless of the initial state, $\kappa$ can be achieved as the (deterministic) policy limit of some TAR policy.  For example, define:
\begin{equation}
\label{eq:Markov_from_kappa}
    p(\ell,x)= \frac{\kappa(\ell,x)}{\kappa(1,x)+\kappa(2,x)}
\end{equation}
for $\ell=1,2$ and $x\in S$. Define $A_n$ to be the following {\em stationary Markov policy}:
\begin{align}
    P(A_n=\ell|\mathcal{G}_{n-1}, X_n)=p(\ell,X_n)
\end{align}
for $\ell=1,2$, $n\ge 0$.  This policy is Markov because it depends only on the current state $X_n$ (and the auxiliary randomness $U_n$) and stationary because the choice probabilities do not change with time.  Further, it is straightforward to check that since each $P(\ell)$ is irreducible for $\ell = 1,2$, this policy makes $X_n$ an irreducible Markov chain.  As a result, this chain therefore has a unique stationary distribution regardless of the initial state.  Elementary computation yields that the stationary distribution must be $\pi(x)=\kappa(0,x)+\kappa(1,x), x\in S$. Therefore, the policy limit of this policy is equal to $\kappa$, regardless of the initial state. 
\end{remark}

The following proposition implies Corollary \ref{cor:consistency}: the MLE estimator $\hat{\alpha}_n$ is {\em consistent} under all TAR policies with  positive policy limits, i.e., $\hat{\alpha}_n$ converges in probability to $\alpha$.  Both proofs are relatively straightforward and provided in Appendix \ref{sec:tar_proofs}.

\begin{proposition}\label{prop: P convergence}
If $A$ is a TAR policy with policy limits that are almost surely  positive, then for $\ell = 1,2$, as $n \to \infty$, there holds
\begin{equation}
\label{eq:P_convergence}
    \hat{P}_n(\ell)\xrightarrow{p}P(\ell)
\end{equation}
and 
\begin{equation}
\label{eq:pi_convergence}
    \hat{\pi}_n(\ell)\xrightarrow{p} \pi(\ell).
\end{equation}
\end{proposition}

\begin{corollary}
\label{cor:consistency}
If $A$ is a TAR policy with policy limits that are almost surely positive, then the MLE estimator $\hat{\alpha}_n$ is {\em consistent} under $A$, i.e., $\hat{\alpha}_n \xrightarrow{p} \alpha$ as $n \to \infty$.
\end{corollary}

\section{The MLE with TAR Policies: A Characterization of Efficiency}
\label{sec:var}

In this section, we study the asymptotic variance of the MLE when TAR policies are used to sample and compare the two Markov chains in the experiment.   In Section \ref{ssec:CLT}, we develop a central limit theorem for the MLE estimator when used with TAR policies.  In Section \ref{ssec:optimality}, we use this central limit theorem to give a characterization of policies that are efficient, in the sense that they provide minimum asymptotic variance.

\subsection{A Central Limit Theorem}
\label{ssec:CLT}

A key tool in our analysis is {\em Poisson's equation} from the theory of Markov chains.  Let $P$ be the transition matrix of an irreducible Markov chain on $S$, with corresponding stationary distribution $\pi$; let $\Pi$ be matrix with rows equal to $\pi$, i.e., $\Pi = e \pi$ where $e= (1, \ldots, 1)$.  Further let $r$ be a reward function on $S$; we center $r$ by defining $\tilde{r} = r - e \pi r$.  Recall that Poisson's equation for $\tilde{r}$ under $P$ is:
\begin{equation}
\label{eq:poisson}
(I - P) g = \tilde{r}.
\end{equation}
One solution to the previous equation is given by:
\begin{equation}
\label{eq:PE_solution}
\tilde{g} = (I - P + \Pi)^{-1} r,
\end{equation}
where $(I - P + \Pi)^{-1}$ is the {\em fundamental matrix} associated to $P$.  (In general, the solution to Poisson's equation is not unique; however, the preceding solution is the unique one for which $\pi \tilde{g} = \pi r$.)  The following result is a well-known central limit theorem for finite Markov chains (see, e.g., \cite{asmussen2003applied}, Theorem 7.2).  

\begin{proposition}
\label{prop:poisson_CLT}
As $n \to \infty$, the random variable $\frac{1}{\sqrt{n}} \left( r(X_0) + \cdots + r(X_{n-1}) - n \pi r \right)$
converges weakly to a normal random variable with mean zero and variance $\sigma^2(r) = \pi \tilde{g}^2 - \pi (P\tilde{g})^2$, where $\tilde{g}$ is the solution to Poisson's equation in \eqref{eq:PE_solution}.\footnote{Here we use the notation $f^2$ to denote elementwise squaring of the function, i.e., $f^2(x) = f(x)^2$.}
\end{proposition}

Our goal is to obtain a central limit theorem for TAR policies.  Note that there are several complexities in our setting that make this challenging: in general TAR policies allow for adaptive sampling, i.e., the chain chosen by the policy at a given time step can depend on the past history.  In particular, the induced state process may no longer be Markovian as a result.  To further complicate matters, the MLE $\hat{\alpha}_n$ cannot be represented simply as an average sum of rewards over the first $n$ time periods. 

Nevertheless, we now present a central limit theorem result analogous to Proposition \ref{prop:poisson_CLT} for MLE estimation with TAR policies. 
For $\ell = 1,2$, we define $\tilde{g}(\ell)$ to be the solution to \eqref{eq:PE_solution} for the transition matrix $P(\ell)$ with reward function $r(\ell)$, i.e.:
\[ \tilde{g}(\ell) = (I - P(\ell) + \Pi(\ell))^{-1} r(\ell). \]
In addition, define:
\begin{align}
\sigma^2(\ell,x) &= \Var\bigg(\tilde{g}(\ell,X_1)+R_1\big|X_{0}=x, A_{0}=\ell\bigg)\\
    &=\sum_{y\in S} P(\ell, x, y) [\tilde{g}(\ell,y)-\sum_{z\in S} P(\ell,x,z)\tilde{g}(\ell,z)]^2\\
    &\quad +\sum_{y\in S} P(\ell,x,y)\Var(R_1|X_0=x, X_1=y, A_0=\ell). \label{eq:sigma}
\end{align}

We have the following theorem.

\begin{theorem}\label{thm:CLT}
Suppose that $A$ is a TAR policy, with policy limits $\gamma=(\gamma(\ell,x): \ell=1,2,x\in S)$ that are almost surely  positive.  Let $G=\big(G(\ell,x): x\in S, \ell=1,2\big)$ be a family of independent Gaussian random variables with mean $0$ and unit variance.  Then for the MLE $\hat{\alpha}_n$, there holds
\begin{align}
\label{eq:CLT}
    n^{1/2}(\hat{\alpha}_n-\alpha)\Rightarrow \sum_{x\in S} \frac{\pi(1,x)\sigma(1,x)}{\gamma(1,x)^{1/2}} G(1,x)-\sum_{x\in S}\frac{\pi(2,x)\sigma(2,x)}{\gamma(2,x)^{1/2}}G(2,x)
\end{align}
as $n\rightarrow \infty$, where $G$ is independent of $\gamma$.
\end{theorem}

The full proof of Theorem \ref{thm:CLT} is in Appendix \ref{sec:var_proofs}.  A key idea in the proof is to show, via Poisson's equation, that:
\[     \hat{\alpha}_n(\ell) - \alpha(\ell) = \big(\hat{\pi}_n(\ell)-\pi(\ell)\big)r(\ell)
    =\hat{\pi}_n(\ell)\big(\hat{P}_n(\ell)-P(\ell)\big)\tilde{g}(\ell).
\]
We are then able to apply martingale arguments to analyze the right hand side of the preceding expression, by looking at the difference in the realized state transition, and the expected state transition.  These steps allow us to leverage classical martingale central limit theorem results; see Appendix \ref{sec:var_proofs} for details.

For later reference, the following two results will be useful.  The first applies Theorem \ref{thm:CLT} to show that we can lower bound the scaled asymptotic variance of the MLE.  The second shows that for TAR policies that have constant and positive policy limits, in fact we can exactly obtain the scaled asymptotic variance of the MLE.  We later use these results to show the existence of an optimal TAR policy with constant policy limits.  Proofs of both results are in the Appendix \ref{sec:var_proofs}.


\begin{corollary}
\label{cor:var_lowerbound}
Let $A$ be a TAR policy with almost surely positive policy limits $\gamma$, with associated MLE $\hat{\alpha}_n$.  
Then there holds:
\begin{equation}
\label{eq:var_lowerbound}
\liminf_{n \to \infty} n \Var(\hat{\alpha}_n - \alpha) \geq \sum_{\ell=1,2} \sum_{x\in S} \frac{\pi^2(\ell,x)\sigma^2(\ell,x)}{E\{\gamma(\ell,x)\}}.
\end{equation}
\end{corollary}

\begin{corollary}
\label{cor:asymptotic_var_constant_TAR}
Let $A$ be a TAR policy with almost surely constant and positive policy limits $\gamma$.  Let $\hat{\alpha}_n$ be the associated MLE.  Then:
\begin{equation}
    \label{eq:const_policy_limit_var}
 \lim_{n \to \infty} n \Var(\hat{\alpha}_n - \alpha) = \sum_{\ell=1,2} \sum_{x\in S} \frac{\pi^2(\ell,x)\sigma^2(\ell,x)}{\gamma(\ell,x)}.
\end{equation}
\end{corollary}

\subsection{Optimal TAR Policies with the MLE Estimator}
\label{ssec:optimality}

In this section we characterize efficient policies, i.e., those that achieve minimum asymptotic variance within the class of TAR policies with almost surely positive policy limits.  (Note that we know such policies are consistent from Corollary \ref{cor:consistency}, and thus we focus solely on asymptotic variance in considering efficiency.)  We show that the policy limits of any such policy can be characterized via the solution to a particular convex optimization problem.  We require the mild additional assumption in this section that $\sigma(\ell,x) > 0$ for all $\ell,x$; this will hold, e.g., if rewards are random in each state.

We start with the following formal definition of efficiency.
\begin{definition}
\label{def:efficient}
Let $A^*$ be a TAR policy with almost surely positive policy limits, with associated MLE $\hat{\alpha}^*_n$.  We say that $A^*$ is {\em efficient} if it has lower (scaled) asymptotic variance than any other TAR policy with almost surely positive policy limits; i.e., for any such policy $A$ with associated MLE $\hat{\alpha}_n$, there holds:
\begin{equation}
\label{eq:efficient}
 \limsup_{n \to \infty} n \Var(\hat{\alpha}^*_n - \alpha) \leq \liminf_{n \to \infty} n \Var(\hat{\alpha}_n - \alpha).
 \end{equation}
\end{definition}

The following theorem leverages the central limit theorem in Theorem \ref{thm:CLT}, and in particular Corollaries \ref{cor:var_lowerbound} and \ref{cor:asymptotic_var_constant_TAR}, to 
give an optimization problem whose solution characterizes the optimal (scaled) asymptotic variance.  The proof of the theorem is in the Appendix.

\begin{theorem}
\label{thm: stationary optimality}
Suppose that for all $\ell,x$, there holds $\sigma(\ell,x) > 0$.  Consider the following (convex) optimization problem:
\begin{align}
    \text{minimize}\ \ \ & \sum_{\ell=1}^{2} \sum_{x\in S} \frac{\pi^2(\ell,x)\sigma^2(\ell,x)}{\kappa(\ell,x)} \label{eq:MLE_opt1}\\
    \text{subject to}\ \ \ & \kappa \in \mathcal{K}. \label{eq:MLE_opt2}
\end{align}
This problem has a unique solution $\kappa^*$, and all entries of $\kappa^*$ are positive.  Any TAR policy $A^*$ that has almost surely constant policy limits $\kappa^*$ is efficient, and the scaled asymptotic variance of the MLE under $A^*$ is given by \eqref{eq:MLE_opt1} evaluated at $\kappa^*$; i.e., if we let $\hat{\alpha}_n^*$ denote the resulting MLE, we have:
\begin{equation}
    \label{eq:Astar_var}
 \lim_{n \to \infty} n \Var(\hat{\alpha}_n^* - \alpha) = \sum_{\ell=1}^{2} \sum_{x\in S} \frac{\pi^2(\ell,x)\sigma^2(\ell,x)}{\kappa^*(\ell,x)}.
\end{equation}
\end{theorem}

\section{An Online Experimental Design: \OnlineETI}
\label{sec:online}

Based on Theorem \ref{thm: stationary optimality}, it follows that one efficient policy is the stationary Markov policy obtained by inserting $\kappa^*$ in \eqref{eq:Markov_from_kappa}.  However, such a policy requires knowledge of the system parameters (as these are required to solve the optimization problem \eqref{eq:MLE_opt1}-\eqref{eq:MLE_opt2} that yields $\kappa^*$).  Of course, if these parameters were already known, there would be no need for experiment design and estimation in the first place.

In this section, we instead construct an {\em online} policy (i.e., one that does not use {\em a priori} knowledge of system parameters) that is consistent and efficient in the limit as $n \to \infty$.  In particular, the policy we construct will be TAR with policy limits $\kappa^*$.  Our proposed policy, called \OnlineETI (for {\em Online Experimentation with Temporal Interference}) works as follows.  At every time step $n$, \OnlineETI maintains an MLE of $P(\ell)$ as $\hat{P}_n(\ell)$.  Initially, in every state, the policy samples chain $\ell = 1,2$ with probability 0.5; this continues until $\hat{P}_n(\ell)$ becomes irreducible (with the associated stationary distribution denoted $\hat{\pi}_n(\ell)$).  \OnlineETI estimates the mean and variance of one-step rewards as well, for each triple $(A_n, X_n, X_{n+1}) = (\ell,x,y)$, $\ell = 1,2$, $x,y \in S$.  These estimates are used to estimate $\tilde{g}(\ell)$, and thus yield an estimate $\hat{\sigma}_n(\ell,x)$ for $\sigma(\ell,x)$ (cf.~\eqref{eq:sigma}).  

Using these estimates, \OnlineETI computes $\hat{\kappa}_n$ as the minimizer of $\sum_{\ell = 1,2} \sum_{x \in S} \hat{\pi}_n(\ell,x)^2 \hat{\sigma}_n(\ell,x)^2/\kappa(\ell,x)$ over $\kappa \in \mathcal{K}$ (cf.~\eqref{eq:MLE_opt1}-\eqref{eq:MLE_opt2}).  At each time step $n$, \OnlineETI then samples from chain $\ell$ with the following probability:
\begin{equation}
\label{eq:p_n(ell,x)}
  \hat{p}_n(\ell,x) = \left(1-M_n(x)^{-1/2}\right)\left(\frac{\hat{\kappa}_n(\ell,x)}{\hat{\kappa}_n(1,x)+\hat{\kappa}_n(2,x)}\right) +\frac{1}{2}M_n(x)^{-1/2}, 
\end{equation}
where $M_n(x)$ is the cumulative number of visits to state $x$ up to time $n$.  In other words, $\hat{p}_n(\ell,x)$ is chosen proportional to $\hat{\kappa}_n(\ell,x)$, but with an additional $1/2 \cdot M_n(x)^{-1/2}$ probability of playing either chain.  This latter term is {\em forced exploration}: it ensures sufficient exploration of both chains to give consistent estimates of model parameters, without influencing the policy limits, and thus the asymptotic variance, of the policy.  As $n \to \infty$, we show that $\hat{\kappa}_n(\ell,x) \to \kappa^*(\ell,x)$ almost surely, and thus that $\hat{p}_n(\ell,x)$ converges to the choice probability of the optimal stationary Markov policy defined by inserting $\kappa^*$ in \eqref{eq:Markov_from_kappa}.  The full pseudocode of $\OnlineETI$ is in Appendix \ref{sec:online_pseudocode}.  We have the following theorem; the full proof appears in Appendix \ref{sec:proofs_online}.


\begin{theorem}\label{thm: alg efficiency} Suppose that for all $\ell,x$, there holds $\sigma(\ell,x) > 0$. Then \OnlineETI  is a TAR policy with policy limits $\kappa^*$ (as defined in Theorem \ref{thm: stationary optimality}), and thus it is consistent and efficient.
\end{theorem}

\section{Conclusion}
\label{sec:conclusion}

We close with a few comments about additional results and promising future directions.  First, we note that in our paper, we have focused on an adaptive (randomized) experimental design that can switch between the two Markov chains at any state, and we use nonparametric maximum likelihood estimation.  An alternative approach involves estimation via {\em sample averages}: for any time horizon $T$, we compute the time-average reward earned by chains $\ell = 1,2$ respectively, during each of the periods where each corresponding chain was run.  Indeed, standard switchback designs (cf.~Section \ref{sec:intro}) correspond to such designs, but where the chain being run is changed at predetermined time intervals.

In general, sample average estimation (SAE) will not be consistent, as discussed in the Introduction.  However, one approach to eliminating this asymptotic bias is to use a {\em regenerative} aproach: we fix a single state $x_0 \in S$, and {\em only} consider policies that change chains when in the state $x_0$.  Specifically, we consider regenerative policies that choose $\ell = 1,2$ with probabilities $q_1, q_2 = 1 - q_1$ when in state $x_0$, respectively (where $0 < q_1, q_2 < 1$).  We refer to these designs as {\em regenerative} policies.  By the strong Markov property, SAE with regenerative policies will yield a consistent (i.e., asymptotically unbiased) estimate of the treatment effect $\alpha$.  Using similar techniques to those described in this paper, it is straightforward to characterize the asymptotic variance of such a policy under SAE, and again the variance-minimizing regenerative policy can be obtained in terms of model primitives.  See \cite{johari2019simons} for further details.

Second, in this paper we focus on asymptotic efficiency rather than fixed finite horizons.  In fact, the term $M_n(x)^{-1/2}$ in \eqref{eq:p_n(ell,x)} can be replaced by $M_n(x)^{-\beta}$ for any $0 < \beta < 1$, and the result of Theorem \ref{thm: alg efficiency} still holds.  We conjecture that $\beta$ serves as a tuning parameter between finite horizon bias and variance: informally, as $\beta$ increases, forced exploration decreases, yielding higher bias but lower variance.  A formal investigation of optimal finite horizon experimental design remains an important direction for future work.

Finally, our approach is fully nonparametric, with no structural assumptions on the system.  In practice, full knowledge of the state space may be infeasible.  Nevertheless, we believe our main insights regarding experimental design that adaptively samples the two chains, with the goal of minimizing variance of the treatment effect estimate, can still provide guidance on optimal design in such applied domains.  This remains an active direction of future study.






\bibliographystyle{abbrv}
\bibliography{myrefs}

\begin{thebibliography}{10}

\bibitem{asmussen2003applied}
S.~Asmussen.
\newblock {\em Applied Probability and Queues}.
\newblock Springer, 2003.

\bibitem{athey2018exact}
S.~Athey, D.~Eckles, and G.~W. Imbens.
\newblock Exact p-values for network interference.
\newblock {\em Journal of the American Statistical Association},
  113(521):230--240, 2018.

\bibitem{basse2016randomization}
G.~W. Basse, H.~A. Soufiani, and D.~Lambert.
\newblock Randomization and the pernicious effects of limited budgets on
  auction experiments.
\newblock In {\em Artificial Intelligence and Statistics}, pages 1412--1420,
  2016.

\bibitem{brandt1938tests}
A.~Brandt.
\newblock Tests of significance in reversal or switchback trials.
\newblock {\em Research Bulletin (Iowa Agriculture and Home Economics
  Experiment Station)}, 21(234):1, 1938.

\bibitem{Lyft}
N.~Chamandy.
\newblock Experimentation in a ridesharing marketplace, 2016.

\bibitem{dann2017unifying}
C.~Dann, T.~Lattimore, and E.~Brunskill.
\newblock Unifying pac and regret: Uniform pac bounds for episodic
  reinforcement learning.
\newblock In {\em Advances in Neural Information Processing Systems}, pages
  5713--5723, 2017.

\bibitem{di2012policy}
D.~Di~Castro, A.~Tamar, and S.~Mannor.
\newblock Policy gradients with variance related risk criteria.
\newblock {\em arXiv preprint arXiv:1206.6404}, 2012.

\bibitem{eckles2017design}
D.~Eckles, B.~Karrer, and J.~Ugander.
\newblock Design and analysis of experiments in networks: Reducing bias from
  interference.
\newblock {\em Journal of Causal Inference}, 5(1), 2017.

\bibitem{filar1989variance}
J.~A. Filar, L.~C. Kallenberg, and H.-M. Lee.
\newblock Variance-penalized markov decision processes.
\newblock {\em Mathematics of Operations Research}, 14(1):147--161, 1989.

\bibitem{hall1980martingale}
P.~Hall and C.~C. Heyde.
\newblock {\em Martingale limit theory and its application}.
\newblock Academic press, 1980.

\bibitem{hudgens2008toward}
M.~G. Hudgens and M.~E. Halloran.
\newblock Toward causal inference with interference.
\newblock {\em Journal of the American Statistical Association},
  103(482):832--842, 2008.

\bibitem{iancu2015tight}
D.~A. Iancu, M.~Petrik, and D.~Subramanian.
\newblock Tight approximations of dynamic risk measures.
\newblock {\em Mathematics of Operations Research}, 40(3):655--682, 2015.

\bibitem{johari2019simons}
R.~Johari.
\newblock Temporal interference and the design of switchback experiments: A
  markov chain approach, 2019.
\newblock Presentation at the Simons Institute for the Theory of Computing.
  \url{https://simons.berkeley.edu/talks/temporal-interference-design-switchback-experiments-markov-chain-approach}.

\bibitem{johari2020twosided}
R.~Johari, H.~Li, and G.~Weintraub.
\newblock Experimental design in two-sided platforms: An analysis of bias.
\newblock {\em arXiv preprint arXiv:2002.05670}, 2020.
\newblock "To appear in ACM Economics and Computation 2020".

\bibitem{Doordash}
D.~Kastelman.
\newblock Switchback tests and randomized experimentation under network effects
  at {DoorDash}, 2018.

\bibitem{kohavi2009controlled}
R.~Kohavi, R.~Longbotham, D.~Sommerfield, and R.~M. Henne.
\newblock Controlled experiments on the web: survey and practical guide.
\newblock {\em Data mining and knowledge discovery}, 18(1):140--181, 2009.

\bibitem{lucas1956switchback}
H.~Lucas.
\newblock Switchback trials for more than two treatments.
\newblock {\em Journal of Dairy Science}, 39(2):146--154, 1956.

\bibitem{mannor2011mean}
S.~Mannor and J.~Tsitsiklis.
\newblock Mean-variance optimization in markov decision processes.
\newblock {\em arXiv preprint arXiv:1104.5601}, 2011.

\bibitem{manski2013identification}
C.~F. Manski.
\newblock Identification of treatment response with social interactions.
\newblock {\em The Econometrics Journal}, 16(1):S1--S23, 2013.

\bibitem{ostrovsky2011reserve}
M.~Ostrovsky and M.~Schwarz.
\newblock Reserve prices in internet advertising auctions: a field experiment.
\newblock {\em EC}, 11:59--60, 2011.

\bibitem{putta2017pure}
S.~R. Putta and T.~Tulabandhula.
\newblock Pure exploration in episodic fixed-horizon markov decision processes.
\newblock In {\em AAMAS}, pages 1703--1704, 2017.

\bibitem{sobel2006randomized}
M.~E. Sobel.
\newblock What do randomized studies of housing mobility demonstrate? causal
  inference in the face of interference.
\newblock {\em Journal of the American Statistical Association},
  101(476):1398--1407, 2006.

\bibitem{sobel1982variance}
M.~J. Sobel.
\newblock The variance of discounted markov decision processes.
\newblock {\em Journal of Applied Probability}, 19(4):794--802, 1982.

\bibitem{sobel1994mean}
M.~J. Sobel.
\newblock Mean-variance tradeoffs in an undiscounted mdp.
\newblock {\em Operations Research}, 42(1):175--183, 1994.

\bibitem{wager2019experimenting}
S.~Wager and K.~Xu.
\newblock Experimenting in equilibrium.
\newblock {\em arXiv preprint arXiv:1903.02124}, 2019.

\bibitem{yu2018approximate}
P.~Yu, W.~B. Haskell, and H.~Xu.
\newblock Approximate value iteration for risk-aware markov decision processes.
\newblock {\em IEEE Transactions on Automatic Control}, 63(9):3135--3142, 2018.

\end{thebibliography}

\newpage
\appendix

\section{An Example: Cooperative Exploration}
\label{sec:example_cooperative}

Throughout this section, we refer to the two Markov chains depicted in Figure \ref{fig: cooperative}.  The state space for both chains is $S = \{ 1, \ldots, s\}$, where $s > 1$.  The red chain corresponds to $\ell = 1$ and the blue chain corresponds to $\ell = 2$.  The transition probabilities are as depicted in the figure.  In particular, we assume that chain $1$ has $P(x,x+1) = P(s,1) = 1$  for $x = 1,\ldots,s-1$, and chain $2$ has $P(x,x-1) = P(1,s) = 1$ for $x = 2,\ldots,s$.  

We assume the experimenter {\em knows} the transition matrices exactly (as they are deterministic), and thus the only uncertainty in estimating the reward distribution comes from uncertainty regarding the reward distribution of each chain.  We assume each chain {\em only} earns a reward in state $x = 1$.  In particular, chain $\ell$ earns a reward that is Bernoulli($q(\ell)$) in state $1$, for some unknown parameter $q(\ell)$ with $0 < q(\ell) < 1$.  Clearly the stationary distribution of each chain is $\pi(\ell,x) = 1/s$, and so the steady state mean reward of each chain is $\alpha(\ell) = q(\ell)/s$.  Thus the treatment effect is $(q(2) - q(1))/s$.

First, suppose that for $\ell = 1,2$ we wanted to estimate only $\alpha(\ell)$ by running chain $\ell$, i.e., $A_n = \ell$ for all $n$.  Then note that in every $S$ steps, only one observation is received of the reward in state $1$.  Let $\hat{\alpha}_n(\ell)$ denote the maximum likelihood estimate of steady state reward obtained from the first $n$ steps.  Given the structure of this chain, it is straightforward to check that the MLE at time $n > s$ reduces to the sample average of $\lfloor n/s \rfloor$ independent Bernoulli($q(\ell)$) samples.  This estimator has variance that scales as $\Theta(s/n)$.  Thus, any attempt at estimation of the variance of steady state reward by running each chain in isolation will have variance that scales with $s$.

On the other hand, now suppose we use the following sampling policy: the policy always samples chain $1$ when in state $s$;  the policy always samples chain $2$ in states $2, \ldots, s-1$; and in successive visits to state $1$, the policy deterministically alternates between sampling chains $1$ and $2$.
Suppose for simplicity that this chain starts at $X_0 = 1$.  Then in every four periods, this chain obtains one independent sample each of a reward from chain 1 in state 1 (i.e., Bernoulli($q(1)$), and a reward from chain 2 in state 1 (i.e., Bernoulli($q(2)$).  Thus the maximum likelihood estimator of $\alpha(\ell)$ will have variance that scales as $\Theta(4/n)$, and in particular, does {\em not} grow with $s$.  In particular, the improvement in variance under this policy relative to the preceding approach can be made unboundedly large by increasing $s$.

This example illustrates the surprising insight that by {\em cooperatively exploring} using {\em both} chains together, substantial benefits in estimation variance can be achieved relative to the variance of estimation with each chain in isolation.  In this example, both approaches to estimation will be consistent.
However, the state-dependent sampling policy leads to a substantial reduction in variance, because it benefits from cooperative exploration: for each chain $\ell =1,2$, the {\em other} chain is used to drive the system back to where samples are most needed to reduce variance.  By contrast, running each chain in isolation forces the experimenter to wait $s$ time steps between successive observations of the random reward in state $1$.  When $s$ becomes larger, the long run average time spent in state 1 approaches $1/2$ for the state-dependent sampling policy, but approaches zero for either chain in isolation. 

\begin{figure}[h]
    \centering
    \includegraphics[width=0.3\textwidth]{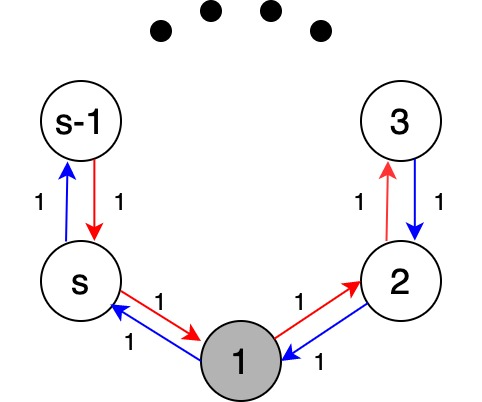}
    \caption{The two Markov chains described in Appendix \ref{sec:example_cooperative}.  Chain $1$ is red, and chain $2$ is blue.  Rewards are only earned in state $1$ for each chain; in particular, the reward distribution in state 1 is Bernoulli($q(\ell)$) for chain $\ell$.}
    \label{fig: cooperative}
\end{figure}

\section{Proofs: Section \ref{sec:tar}}
\label{sec:tar_proofs}

\begin{proof}[Proof of Proposition \ref{prop: policy limits}]
Relations (\ref{eq: kappa eq 2}) and (\ref{eq: kappa eq 3}) are obvious. As for (\ref{eq: kappa eq 1}), note that
\begin{align}
&\frac{1}{n}\Gamma_n(1,y)+\frac{1}{n} \Gamma_n(2,y)\\
=&\frac{1}{n}\sum_{j=\newrj{0}}^{n-1}I(X_j=y)\\
=&\frac{1}{n}\sum_{j=1}^{n}I(X_j=y)+\frac{1}{n}\big(I(X_0=y)-I(X_n=y)\big)\\
=& \frac{1}{n}\sum_{j=1}^{n}I(X_j=y)+O_p\left(\frac{1}{n}\right)\\
=& \sum_{\ell=1}^{2} \sum_{x\in S} \frac{1}{n} \sum_{j=1}^{n} I(X_{j-1}=x, A_{j-1}=\ell, X_j=y)+O_p\left(\frac{1}{n}\right)\\
=& \sum_{\ell=1}^{2}\sum_{x\in S} \frac{1}{n}\Gamma_n(\ell,x)P(\ell,x,y)\\
&+\sum_{\ell=1}^{2}\sum_{x\in S}\bigg\{ \frac{1}{n} \sum_{j=1}^{n} I(X_{j-1}=x, A_{j-1}=\ell)[I(X_j=y)-P(\ell,x,y)]\bigg\}+O_p\left(\frac{1}{n}\right).\label{eq: gamma eq 4}
\end{align}
(Here we use the notation that $f(n) =O_p(1/n)$ to denote stochastic boundedness of $n f(n)$: for all $\epsilon > 0$, there exists deterministic $M$ such that $P(|n f(n)| > M) < \epsilon$ for all $n$.)

Let $W_j = I(X_{j-1}=x, A_{j-1}=\ell)[I(X_j=y)-P(\ell,x,y)]$.  This is a martingale difference sequence adapted to  $\mathcal{G}_{j}$.  In particular, as a result the $W_j$ are {\em orthogonal} in the sense that for $j < k$, there $E\{_j W_k\} = 0$.  (This result follows by conditioning on $\mathcal{G}_j$ and nesting conditional expectations: $E\{W_j W_k\} = E\{E\{W_k | \mathcal{G}_j\} W_j\} = 0$.)  Using orthogonality of the martingale differences implies that 
\begin{align}
    &E\bigg\{\bigg(\frac{1}{n}\sum_{j=1}^{n} I(X_{j-1}=x, A_{j-1}=\ell)[I(X_j=y)-P(\ell,x,y)]\bigg)^2\bigg\}\\
    =&E\bigg\{\frac{1}{n^2} \Gamma_{n}(\ell,x)\big(P(\ell,x,y)-P^2(\ell,x,y)\big)\bigg\} \rightarrow  0
\end{align}
as $n\rightarrow \infty$. Therefore, by Chebyshev's inequality
\begin{align}\label{eq: convergence P(l) martingale}
    \frac{1}{n}\sum_{j=1}^{n} I(X_{j-1}=x, A_{j-1}=\ell)[I(X_j=y)-P(\ell,x,y)]\xrightarrow{p} 0
\end{align}
as $n\rightarrow \infty$. Taking the limits in (\ref{eq: gamma eq 4}) yields (\ref{eq: kappa eq 1}). 
\end{proof}

\begin{proof}[Proof of Proposition \ref{prop: P convergence}]
We start by proving \eqref{eq:P_convergence}.  We recall that 
\begin{align}
    \hat{P}_n(\ell,x,y)=\frac{\sum_{j=0}^{n-1} I(X_j=x, A_j=\ell, X_{j+1}=y)}{\max\{\Gamma_n(\ell,x), 1\}}
\end{align}
As in the proof of Proposition \ref{prop: policy limits}, 
\begin{align}
    &\frac{1}{n}\sum_{j=0}^{n-1}I(X_j=x, A_j=\ell, X_{j+1}=y)\\
    =& \bigg\{\frac{1}{n}\sum_{j=0}^{n-1}I(X_j=x, A_j=\ell)[I(X_{j+1}=y)-P(\ell,x,y)]\bigg\}+\frac{1}{n}\Gamma_n(\ell,x)P(\ell,x,y)
\end{align}
Therefore,
\begin{align}
   \hat{P}_n(\ell,x,y)=&\frac{\bigg\{\frac{1}{n}\sum_{j=0}^{n-1}I(X_j=x, A_j=\ell)[I(X_{j+1}=y)-P(\ell,x,y)]\bigg\}}{\frac{1}{n}\max\{\Gamma_n(\ell,x),1\}} 
   \\ & +\frac{\Gamma_n(\ell,x)}{\max\{\Gamma_n(\ell,x),1\}}P(\ell,x,y)\\
    =& \frac{o_p(1)}{\frac{1}{n}\max\{\Gamma_n(\ell,x),1\}}+\frac{\Gamma_n(\ell,x)}{\max\{\Gamma_n(\ell,x),1\}}P(\ell,x,y) \qquad \text{from (\ref{eq: convergence P(l) martingale})}\\
    \xrightarrow{p}& P(\ell,x,y)\label{eq: convergence to P}
\end{align}
as $n\rightarrow \infty$, where the convergence in (\ref{eq: convergence to P}) holds 
because $\gamma(\ell,x)$ are almost surely positive.

We now prove \eqref{eq:pi_convergence}.  
Let $\mu_n$ denote the law of $\hat{\pi}_n$, and view it as a probability measure on vectors in the probability simplex on the state space $S$, denoted $\Delta(S)$.  The set $\Delta(S)$ is compact, and so by Prohorov's theorem there exists a deterministic subsequence $n_k$ such that $\mu_{n_k}$ converges weakly to a probability measure $\mu$ on $\Delta(S)$, with associated random variable $\pi'(\ell)$.  Since $\hat{P}_n(\ell) \xrightarrow{p} P(\ell)$ by \eqref{eq: convergence to P}, and $P(\ell)$ is deterministic, it follows by Slutsky's theorem that:
\[ \hat{\pi}_{n_k}(\ell) \hat{P}_{n_k}(\ell) \Rightarrow \pi'(\ell) P(\ell). \]

Since the policy limits are almost surely positive, $J$ is almost surely finite.  Thus, for all sufficiently large $k$ there holds $\hat{\pi}_{n_k}(\ell) \hat{P}_{n_k}(\ell) = \hat{\pi}_{n_k}(\ell)$.   It follows that $\pi'(\ell) = \pi'(\ell) P(\ell)$, so that $\pi'(\ell) = \pi(\ell)$ almost surely.  In other words, the measure $\mu$ is the Dirac measure that places probability one on $\pi(\ell)$.  Since this is the case for every convergent subsequence of $\{\mu_n\}$, we conclude that $\hat{\pi}_n(\ell) \Rightarrow \pi(\ell)$.  Since $\pi(\ell)$ is deterministic, we conclude that $\pi_n(\ell) \xrightarrow{p} \pi(\ell)$ as $n \to \infty$, as required.
\end{proof}

\begin{proof}[Proof of Corollary \ref{cor:consistency}]
Since the policy limits of $A$ are almost surely positive, it is straightforward to show that for each $\ell,x$, $\hat{r}_n(\ell,x) \xrightarrow{p} r(\ell,x)$ as $n \to \infty$.  The result then follows from Proposition \ref{prop: P convergence}.
\end{proof}

\section{Proofs: Section \ref{sec:var}}
\label{sec:var_proofs}

\begin{proof}[Proof of Theorem \ref{thm:CLT}]
We begin by showing that for $\ell = 1,2$ and $n \geq J$, there holds:
\begin{equation}
    \big(\hat{\pi}_n(\ell)-\pi(\ell)\big)r(\ell)=\hat{\pi}_n(\ell)\big(\hat{P}_n(\ell)-P(\ell)\big)\tilde{g}(\ell).
\end{equation}
To see this, observe that for $n\ge J$,
\begin{align}
    \hat{\pi}_n(\ell)-\pi(\ell)=\hat{\pi}_n \hat{P}_n(\ell)-\hat{\pi}_n(\ell)P(\ell)+\hat{\pi}_n(\ell)P(\ell)-\pi(\ell)P(\ell)
\end{align}
so rearranging the terms we get,
\begin{align}
    \big(\hat{\pi}_n(\ell)-\pi(\ell)\big)\big(I-P(\ell)\big)=\hat{\pi}_n(\ell)\big(\hat{P}_n(\ell)-P(\ell)\big). 
\end{align}
Because $\Pi(\ell)$ has identical elements down each column, 
\begin{align}
   \big (\hat{\pi}_n(\ell)-\pi(\ell)\big)\Pi(\ell)=0,
\end{align}
and hence
\begin{align}
    \big(\hat{\pi}_n(\ell)-\pi(\ell)\big)\big(I-P(\ell)+\Pi(\ell)\big)=\hat{\pi}_n(\ell)\big(\hat{P}_n(\ell)-P(\ell)\big). 
\end{align}
Recall that we defined $\tilde{g}(\ell) = \big(I-P(\ell)+\Pi(\ell)\big)^{-1}r(\ell)$; thus
\begin{align}
\big(\hat{\pi}_n(\ell)-\pi(\ell)\big)r(\ell)=\hat{\pi}_n(\ell)\big(\hat{P}_n(\ell)-P(\ell)\big)\tilde{g}(\ell),
\end{align}
as desired.

Now for $\ell=1,2$, and $x\in S$, define
\begin{align}
    r(\ell,x,y)=&\int_{\mathbb{R}}zF(dz, x,y,\ell).
\end{align}

Recall that $\hat{\alpha}_n = \hat{\pi}_n(2) \hat{r}_n(2) - \hat{\pi}_n(1) \hat{r}_n(1)$. We can write:
\begin{align}
\hat{\pi}_n(\ell)&\hat{r}_n(\ell)-\pi(\ell)r(\ell)\\
    &=\big(\hat{\pi}_n(\ell)-\pi(\ell)\big)r(\ell)+\hat{\pi}_n(\ell)\big(\hat{r}_n(\ell)-r(\ell)\big)\\
    &=\hat{\pi}_n(\ell)\big(\hat{P}_n(\ell)-P(\ell)\big)\tilde{g}(\ell)+\hat{\pi}_n(\ell)\big(\hat{r}_n(\ell)-r(\ell)\big)\\
    &= \sum_{x\in S} \hat{\pi}_n(\ell,x) \frac{\sum_{j=1}^{n} I(X_{j-1}=x, A_{j-1}=\ell)\big[\tilde{g}(\ell,X_j)-\big(P(\ell)\tilde{g}(\ell)\big)(X_{j-1})]}{\max\{\Gamma_n(\ell,x),1\}}\nonumber\\&+ \sum_{x\in S} \hat{\pi}_n(\ell,x) \frac{\sum_{j=0}^{n-1} I(X_j=x, A_j=\ell)\big(R_{j+1}-r(\ell,x)\big)}{\max\{\Gamma_n(\ell,x),1\}}\\
    &= \sum_{x\in S}\hat{\pi}_n(\ell,x)\frac{\sum_{j=1}^{n}D_j(\ell,x)}{\max\{\Gamma_n(\ell,x),1\}}\label{eq:trans_to_Dj}
\end{align}
where 
\begin{align}
\label{eq:Dj(l,x)}
    D_j(\ell,x):= I(X_{j-1}=x, A_{j-1}=\ell)\big[\tilde{g}(\ell,X_j)-\big(P(\ell)\tilde{g}(\ell)\big)(X_{j-1})+R_j-r(\ell,X_{j-1})\big].
\end{align}
Note that for each $\ell,x$, $D_j(\ell,x)$ is a martingale difference sequence adapted to $\mathcal{G}_j$.

For deterministic $w(\ell)=(w(\ell,x): x\in S), \ell=1,2$, consider
\begin{align}
    T_n&= 
    \frac{1}{\sqrt{n}}\sum_{j=1}^{n} \sum_{\ell=1}^{2} \sum_{x\in S} D_j(\ell,x)w(\ell,x)\\
    &\triangleq \frac{1}{\sqrt{n}} \sum_{j=1}^{n} D_j,
 \end{align}
where
\begin{align}
    D_j=\sum_{\ell=1}^{2} \sum_{x\in S} D_{j}(\ell,x)w(\ell,x).
\end{align}
The $D_j$'s are martingale differences adapted to $(\mathcal{G}_j: j\ge 0)$. Since they are bounded by $2\max\{|\tilde{g}(\ell,x)|: x\in S, \ell=1,2\}<\infty$ (since $r(\ell,x)$ is finite), the following conditional Lindeberg's condition holds (Eq. (3.7) of \cite{hall1980martingale}):
\begin{align}
    \text{for all } \epsilon>0, \quad \sum_{j=1}^{n}\frac{1}{n}E\{D_j^2I(|D_j|> \epsilon)|\mathcal{G}_{j-1}\}\xrightarrow{p} 0.
\end{align}

Furthermore,
\begin{align}\label{eq: D j variance}
    \frac{1}{n} \sum_{j=1}^{n} E\{D_j^2|\mathcal{G}_{j-1}\}&=\frac{1}{n}\sum_{\ell=1}^{2}\sum_{x\in S} \sum_{j=0}^{n-1} I(X_j=x, A_j=\ell)\sigma^2(\ell,x)w^2(\ell,x)\\
    &=\sum_{\ell=1}^{2}\sum_{x\in S} \sigma^2(\ell,x)w^2(\ell,x)\frac{\Gamma_n(\ell,x)}{n}\\
    &\xrightarrow{p} \sum_{\ell=1}^{2} \sum_{x\in S} \sigma^2(\ell,x)w^2(\ell,x)\gamma(\ell,x)\triangleq\eta^2,
\end{align}
since $A$ is assumed to be a TAR policy. We therefore conclude that (by Corollary (3.1) of \cite{hall1980martingale}) 
\begin{align}
    T_n
    &\Rightarrow \sum_{\ell=1}^{2} \sum_{x\in S} \sigma(\ell,x)w(\ell,x)\sqrt{\gamma(\ell,x)}G(\ell,x)\quad (stably)
\end{align}
as $n\rightarrow \infty$.\footnote{If a sequence of random variables $Y_n$ on a probability space $(\Omega, \mathcal{F}, P)$ converges weakly to $Y$, we say the convergence is {\em stable} if for all continuity points $y$ of the cumulative distribution function of $Y$ and for all measurable events $E$, the limit $\lim_{n \to \infty} P(\{Y_n \leq y\} \cap E) = Q_y(E)$ exists, and if $Q_y(E) \to P(E)$ as $y \to \infty$.}

Stable weak convergence implies that the following convergence of characteristic functions holds as well:
\begin{align}
    &E\bigg\{\exp\big(iT_n 
+i\sum_{\ell=1}^2\sum_{x\in S}\tilde{w}(\ell,x)\gamma(\ell,x)\big)\bigg\}\\
    \rightarrow &E\bigg\{\exp\big(i\sum_{\ell=1}^{2}\sum_{x\in S} w(\ell,x)G(\ell,x)\sigma(\ell,x)\sqrt{\gamma(\ell,x)}+i\sum_{\ell=1}^{2} \sum_{x\in S}\tilde{w}(\ell,x)\gamma(\ell,x)\big)\bigg\}
\end{align}
as $n\rightarrow \infty$, for any deterministic choice of $\tilde{w}(\ell)=(\tilde{w}(\ell,x): x\in S), j=1,2$. The Cramer-Wold device therefore implies that
\begin{align}\label{eq: cramer wold weak convergence}
\left(\frac{\sum_{j = 1}^n D_j(\ell,x)}{\sqrt{n}} , \gamma(\ell,x): x\in S, \ell = 1,2\right)
    \Rightarrow \left(\sigma(\ell,x)\sqrt{\gamma(\ell,x)}G(\ell,x), \gamma(\ell,x): x\in S, \ell=1,2\right)
\end{align}
as $n\rightarrow \infty$. Consequently, since the $\gamma(\ell,x)$'s are almost surely  positive, 
\begin{align}
\left(\frac{\sum_{j= 1}^n D_j(\ell,x)}{\sqrt{n}\gamma(\ell,x)}: x\in S, \ell=1,2\right)\Rightarrow   \left(\frac{\sigma(\ell,x)G(\ell,x)}{\sqrt{\gamma(\ell,x)}}: x\in S, \ell=1,2\right)
\end{align}
as $n\rightarrow \infty$. Because $\frac{\Gamma_n(\ell,x)}{n\gamma(\ell,x)}\xrightarrow{p} 1$ as $n\rightarrow \infty$, Slutsky's lemma implies that 
\begin{align}
    &\sqrt{n}\left(\frac{\sum_{j = 1}^n D_j(\ell,x)}{\Gamma_n(\ell,x)}: x\in S, \ell=1,2\right)\\
    \Rightarrow& \bigg(\frac{\sigma(\ell,x)G(\ell,x)}{\sqrt{\gamma(\ell,x)}}: x\in S, \ell=1,2\bigg)\label{eq: Slutsky 1}
\end{align}
as $n\rightarrow \infty$. Finally, Result 2, (\ref{eq: Slutsky 1}), and another application of Slutsky's lemma imply that 
\begin{align}
    &\sqrt{n}\left[\sum_{x\in S} \hat{\pi}_n(1,x)\frac{\sum_{j = 1}^n D_j(1,x)}{\Gamma_n(1,x)}-\sum_{x\in S}\pi_{n}(2,x)\frac{\sum_{j = 1}^n D_j(2,x)}{\Gamma_n(2,x)}\right]\\
    \Rightarrow& \sum_{x\in S}\frac{\pi(1,x)\sigma(1,x)G(1,x)}{\sqrt{\gamma(1,x)}}-\sum_{x\in S}\frac{\pi(2,x)\sigma(2,x)G(2,x)}{\sqrt{\gamma(2,x)}}
\end{align}
as $n\rightarrow \infty$, proving the Theorem.
\end{proof}

\begin{proof}[Proof of Corollary \ref{cor:var_lowerbound}]
Note that the Skorohod representation theorem together with Fatou's lemma applied to \eqref{eq:CLT} yields the following:
\begin{equation}
\label{eq:fatou_var_lowerbound}
    \liminf_{n\rightarrow \infty} n \Var(\hat{\alpha}_n - \alpha)  \ge \sum_{\ell=1,2}\sum_{x\in S}\pi^2(\ell,x)\sigma^2(\ell,x) E\left\{\frac{1}{\gamma(\ell,x)}\right\}.
\end{equation}
Using Jensen's inequality on the right hand side of \eqref{eq:fatou_var_lowerbound}, we obtain the result in \eqref{eq:var_lowerbound}, as required.  (Note that $E\{\gamma(\ell,x)\} > 0$ for all $\ell,x$ since we assumed the policy limits are almost surely positive.)
\end{proof}

\begin{proof}[Proof of Corollary \ref{cor:asymptotic_var_constant_TAR}]

First we show the following limits hold:
\begin{align}
\lim_{n \to \infty}& E \left\{ \sup_{\ell = 1,2; x \in S} \left| \frac{\hat{\pi}_n(\ell,x)}{\max\{\Gamma_n(\ell,x),1\}/n} - \frac{\pi(\ell,x)}{\gamma(\ell,x)} \right| \right\} = 0; \label{eq:pi_over_gamma_limit1}\\
\lim_{n \to \infty}& E \left\{ \sup_{\ell = 1,2; x \in S} \left| \frac{\hat{\pi}_n(\ell,x)}{\max\{\Gamma_n(\ell,x),1\}/n} - \frac{\pi(\ell,x)}{\gamma(\ell,x)} \right|^2 \right\} = 0.\label{eq:pi_over_gamma_limit2}
\end{align}
We know from Proposition \ref{prop: P convergence} that $\hat{\pi}_n \xrightarrow{p} \pi_n$ for all $\ell,x$.  Further, we know from the definition of policy limits that $\Gamma(\ell,x)/n \xrightarrow{p} \gamma(\ell,x)$ for all $\ell,x$.  Thus the vector $(\hat{\pi}_n, \Gamma_n/n)$ converges in probability to the vector $(\pi, \gamma)$.  Use the Skorohod representation theorem to construct a joint probability space on which these limits hold almost surely.  Then note that each of the terms inside the expectations are bounded in 
\eqref{eq:pi_over_gamma_limit1}-\eqref{eq:pi_over_gamma_limit2}, so the desired results hold by the bounded convergence theorem.

For the next steps, we use the same definitions as in the proof of Theorem \ref{thm:CLT}, and refer the reader there for the relevant notation.  In particular, we define $D_j(\ell,x)$ as in that proof, and use the relationship in \eqref{eq:trans_to_Dj}.  We make the following two definitions:
\begin{align*}
 Y_n(\ell) &= \hat{\pi}_n(\ell) \hat{r}_n(\ell) - \pi(\ell)r(\ell) = \sum_{j = 1}^n \sum_{x \in S} \frac{\hat{\pi}_n(\ell,x)}{\max\{\Gamma_n(\ell,x),1\}/n} \cdot \frac{D_j(\ell,x)}{n};\\
 Z_n(\ell) &= \sum_{j = 1}^n \sum_{x \in S} \frac{\pi(\ell,x)}{\gamma(\ell,x)} \cdot \frac{D_j(\ell,x)}{n}.
\end{align*}

Note that $\hat{\alpha}_n - \alpha = Y_n(2) - Y_n(1)$.  The main remaining step in our proof is to show that we can compute the scaled asymptotic variance of $Z_n(2) - Z_n(1)$, and to use this to upper bound the scaled asymptotic variance of $Y_n(2) - Y_n(1)$.

We now show the following limit holds:
\begin{equation}
\label{eq:Zn_var_limit}
\lim_{n \to \infty} \Var(\sqrt{n} (Z_n(2) - Z_n(1))) = \sum_{x \in S} \frac{\pi^2(\ell,x) \sigma^2(\ell,x)}{\gamma(\ell,x)}.
\end{equation}
Observe that $Z_n(\ell)$ is a weighted sum of martingale differences; thus we use orthogonality of martingale differences again.  In particular, $E\{Z_n(\ell)\} = 0$ for all $n$.  Thus $\Var(\sqrt{n} (Z_n(2) - Z_n(1))) = n E\{ (Z_n(2) - Z_n(1))^2 \}$.  Observe that:
\[ Z_n(1) Z_n(2) = \sum_{j = 1}^n \sum_{k = 1}^n \sum_{x,y \in S} \frac{\pi(1,x) \pi(2,y)}{\gamma(1,x)\gamma(2,y)} \frac{D_j(1,x)D_k(2,y)}{n}. \]
We show that $E\{D_j(1,x) D_k(2,y) \} = 0$.  If $j = k$, then the product $D_j(1,x) D_j(2,y) = 0$ since only one of the two chains $\ell = 1,2$ can be run at time $k$.  If $j > k$, then the tower property of conditional expectations is applied  as usual to give:
\[ E\{ E\{ D_j(1,x) | \mathcal{G}_k \} D_k(2,x) \} = 0. \]
The same holds of course if $j < k$.  Thus we have $E\{Z_n(1) Z_n(2)\} = 0$ for all $n$.  Finally, using \eqref{eq: D j variance} with $w(1,x) = \pi(1,x)/\gamma(1,x)$ and $w(2,x) = 0$, together with the Skorohod representation theorem and the bounded convergence theorem, it follows that:
\[ E \{ n Z_n(1)^2 \} \to \sum_{x \in S} \frac{\pi^2(1,x) \sigma^2(1,x)}{\gamma(1,x)}. \]
(Use of bounded convergence here requires assuming boundedness of rewards.)  An analogous result holds for the limit of $E \{ n Z_n(2)^2\}$.  Combining these steps, we obtain \eqref{eq:Zn_var_limit}.  

Finally, we can establish the following upper bound:
\begin{equation}
\label{eq:var_limsup}
\limsup_{n \to \infty} n \Var(\hat{\alpha}_n - \alpha_n) \leq \sum_{\ell = 1,2} \sum_{x \in S} \frac{\pi^2(\ell,x) \sigma^2(\ell,x)}{\gamma(\ell,x)}.
\end{equation}

To prove this we upper bound the variance of $Y_n(2) - Y_n(1)$ in terms of the variance of $Z_n(2) - Z_n(1)$.  Note that $\Var(Y_n(2) - Y_n(1)) \leq E\{ (Y_n(2) - Y_n(1))^2 \}$.  Further, because $D_j(\ell,x)$ are bounded, there exist constants $M_1,M_2$ such that:
\begin{multline*}
(Y_n(2) - Y_n(1))^2 \leq (Z_n(2) - Z_n(1))^2 +M_1 \sup_{\ell = 1,2; x \in S} \left| \frac{\hat{\pi}_n(\ell,x)}{\max\{\Gamma_n(\ell,x),1\}/n} - \frac{\pi(\ell,x)}{\gamma(\ell,x)} \right| \\
+  M_2 \sup_{\ell = 1,2; x \in S} \left| \frac{\hat{\pi}_n(\ell,x)}{\max\{\Gamma_n(\ell,x),1\}/n} - \frac{\pi(\ell,x)}{\gamma(\ell,x)} \right|^2.
\end{multline*}
Taking expectations on both sides, and applying Steps 2 and 3, establishes \eqref{eq:var_limsup}.  Combining \eqref{eq:var_limsup} with \eqref{eq:var_lowerbound} yields the desired result (note that $E\{\gamma(\ell,x)\} = \gamma(\ell,x)$ since the policy limits are almost surely constant).
\end{proof}

\begin{proof}[Proof of Theorem \ref{thm: stationary optimality}]
First, we show that \eqref{eq:MLE_opt1}-\eqref{eq:MLE_opt2} has a unique optimal solution $\kappa^*$, with entries that are all positive.  It is straightforward to see that the solution to this problem will be  positive in all coordinates, since the objective function approaches infinity as any $\kappa(\ell,x)$ approaches zero (as all $\sigma(\ell,x)$ are positive).  Further, note that the objective function is strictly convex and $\mathcal{K}$ is convex and compact, and thus there must be a unique solution $\kappa^* \in \mathcal{K}$ to the optimization problem \eqref{eq:MLE_opt1}-\eqref{eq:MLE_opt2}.

Next, we show that the limit inferior of the scaled asymptotic variance of the MLE under any TAR policy with positive policy limits is bounded below by the optimal value of \eqref{eq:MLE_opt1}-\eqref{eq:MLE_opt2}.  This follows by applying Corollary \ref{cor:var_lowerbound}.  In particular, from Remark \ref{rem:markov_TAR}, we know $\gamma$ is a probability measure over the set $\mathcal{K}$ (cf.~Definition \ref{def: mathcal K}).  The set $\mathcal{K}$ is convex and compact, and so $\kappa=E\{\gamma\}\in \mathcal{K}$. In particular, as a consequence by applying \eqref{eq:var_lowerbound} we conclude that the optimal value of \eqref{eq:MLE_opt1}-\eqref{eq:MLE_opt2} is a lower bound to $\liminf_{n \to \infty} \Var(\hat{\alpha}_n - \alpha)$.  

Finally, the fact that \eqref{eq:Astar_var} holds follows from Corollary \ref{cor:asymptotic_var_constant_TAR}.  The stationary Markov policy $A^*$ defined via \eqref{eq:Markov_from_kappa} has the constant policy limit $\kappa^*$ (cf.~Remark \ref{rem:markov_TAR}), so it is efficient.  The theorem follows.
\end{proof}

\section{Pseudocode for \OnlineETI}
\label{sec:online_pseudocode}

The pseudocode for \OnlineETI is prsented as Algorithm \ref{alg: reg cycle}.

\begin{algorithm}[ht]
\caption{\OnlineETI (Online Experimentation with Temporal Interference)}\label{alg: reg cycle}
\begin{algorithmic}[1]

\Procedure{Experiment}{initial state $x_0$}
\State Set initial state $X_0 = x_0$
\State Initialization: For $\ell = 1,2$, $x,y \in S$, set ${\hat{P}}_0(\ell,x,y)=\frac{1}{|S|}$; $\Gamma_0(\ell,x)=0$; $\Phi_0(\ell,x,y)=0$;
\State \indent $\Theta_0(\ell,x)=0$; $\Psi_0(\ell,x)=0$; $\Upsilon_0(\ell,x,y)=0$; $\hat{r}_0(\ell,x)=0$; $\hat{s}_0(\ell,x,y)=0$;
\State \indent $\hat{t}_0(\ell,x,y)=0$; $\hat{\pi}_0(\ell, x)=0$; $\hat{p}_0(\ell,x)=0.5$
\For{$n = 1, 2, \ldots$}
\State Set $A_{n-1} = \ell$ with probability $\hat{p}_{n-1}(\ell,x)$, i.e.:
\State \indent $A_{n-1} = 1$ if $U_{n-1} \leq \hat{p}_{n-1}(1,x)$, and $A_{n-1} = 2$ otherwise
\State Run chain $A_{n-1}$, and obtain reward $R_n$ and new state $X_n$
\State For all $\ell = 1,2$, $x,y \in S$:
\State \indent $\Gamma_n(\ell,x)\gets \Gamma_{n-1}(\ell,x)+I(X_{n-1}=x, A_{n-1}=\ell)$
\State \indent $\Phi_n(\ell,x, y)\gets  \Phi_{n-1}(\ell,x, y)+I(X_{n-1}=x, X_{n}=y, A_{n-1}=\ell)$ 
\State \indent $\Theta_n(\ell,x)\gets \Theta_{n-1}(\ell,x)+I(X_{n-1}=x, A_{n-1}=\ell)R_{n}$ 
\State \indent $\Psi_n(\ell,x,y)=\Psi_{n-1}(\ell,x,y)+I(X_{n-1}=x, X_{n}=y, A_{n-1}=\ell)R_{n}$ 
\State \indent $\Upsilon_n(\ell,x,y)\gets \Upsilon_{n-1}(\ell,x,y)+I(X_{n-1}=x, X_{n}=y, A_{n-1}=\ell)R_{n}^2$ 
\State \indent $\hat{P}_n(\ell,x,y)\gets \frac{\Phi_n(\ell,x,y)}{\max\{\Gamma_n(\ell,x), 1\}}$
\If{for both $\ell = 1,2$, $\hat{P}_n(\ell)$ is irreducible}
\State Set $\hat{\pi}_n(\ell)$ to be the unique steady state distribution of $\hat{P}_n(\ell)$
\State For $\ell =1,2$ and $x,y \in S$:
\State \indent $\hat{\Pi}_n(\ell)\gets e\hat{\pi}_n(\ell)$
\State \indent $\hat{\tilde{g}}_n(\ell,x)\gets \big(I-\hat{P}_n(\ell) + \hat{\Pi}_n(\ell)\big)^{-1}\hat{r}_n(\ell)$
\State \indent$\hat{r}_n(\ell,x)\gets \frac{\sum_{y\in S}\Psi_n(\ell,x,y)}{\max\{\Gamma_n(\ell,x), 1\}}$
\State \indent$\hat{s}_n(\ell,x,y)\gets \frac{\Psi_n(\ell,x,y)}{\max\{\Phi_n(\ell,x, y), 1\}}$
\State \indent$\hat{t}_n(\ell,x,y)\gets \frac{\Upsilon_n(\ell,x,y)}{\max\{\Phi_n(\ell,x, y), 1\}}$
\State \indent$\hat{\sigma}_n^2(\ell,x)\gets \sum_{y\in S} \hat{P}_n(\ell,x,y)[\hat{\tilde{g}}_n(\ell,y)$ $-$ $\sum_{z\in S}\hat{P}_n(\ell,x,z)\hat{\tilde{g}}_n(\ell,z)]^2$
\State \indent\ \  $+$ $\sum_{y\in S}\hat{P}_n(\ell,x,y)\big(\hat{t}_n(\ell,x,y)-\hat{s}_n(\ell,x,y)^2\big)$
\State Choose any $\hat{\kappa}_n$ in $\arg\inf_{\hat{\kappa}\in \mathcal{K}}\sum_{\ell=1}^{2} \sum_{x\in S} \frac{\hat{\pi}_n^2(\ell,x)\hat{\sigma}_n^2(\ell,x)}{\hat{\kappa}_n(\ell,x)}$
\State For all $x \in S$, $M_n(x) \gets \Gamma_n(1,x) + \Gamma_n(2,x)$
\If{$\hat{\kappa}_n(1,x) + \hat{\kappa}_n(2,x) > 0$ and $M_n(x) > 0$}
\State $\hat{p}_n(\ell,x) \gets (1-M_n(x)^{-1/2})\left(\frac{\hat{\kappa}_n(\ell,x)}{\hat{\kappa}_n(1,x)+\hat{\kappa}_n(2,x)}\right)$
\State \indent $+\frac{1}{2}M_n(x)^{-1/2}$ for $\ell = 1,2$, $x \in S$
\Else
\State $\hat{p}_n(\ell,x) = 0.5$ for $\ell = 1,2$, $x \in S$
\EndIf
\State $\hat{\alpha}_n \gets \hat{\pi}_n(2)\hat{r}_n(2)-\hat{\pi}_n(1)\hat{r}_n(1)$
\Else
\State $\hat{p}_n(\ell,x) \gets 0.5$
\State $\hat{\alpha}_n \gets 0$
\EndIf
\EndFor
\EndProcedure
\end{algorithmic}
\end{algorithm}

\section{Proofs: Section \ref{sec:online}}
\label{sec:proofs_online}

\begin{proof}[Proof of Theorem \ref{thm: alg efficiency}]
We establish that for \OnlineETI there holds:
\begin{align}
    \frac{1}{n} \Gamma_n(\ell,x)\xrightarrow{p} \kappa^*(\ell, x),
\end{align}
where $\kappa^*$ is the solution to \eqref{eq:MLE_opt1}-\eqref{eq:MLE_opt2}.

First, note that the forced exploration (i.e., the $M_n(x)^{-1/2}$ term in the definition of $\hat{p}_n(\ell,x)$) ensures that $\Gamma_n(\ell,x) \to \infty$ almost surely for all $\ell,x$.  To see this, note first that as long as $M_n(x) \to \infty$ almost surely, it must be the case that $\Gamma_n(\ell,x) \to \infty$ for $\ell = 1,2$ almost surely as well, due to the forced exploration term, the fact that $\sum_{k \geq 1} k^{-1/2}$ diverges, and the Borel-Cantelli Lemma.  Since the state space is finite, almost surely, there exists at least one state $x'$ that is visited infinitely often.  Thus almost surely, all states reachable from $x'$ in one step under either $P(1)$ or $P(2)$ must be visited infinitely often as well.  The same argument applies to those states, and so on.  Since the state space is finite, and both $P(1)$ and $P(2)$ are irreducible, this process exhausts all the states, and we conclude $M_n(x) \to \infty$ almost surely for all $x \in S$.

Next we show that for all $\ell,x,y$, $\hat{P}_n(\ell,x,y)$ converges to $P(\ell,x,y)$ almost surely.  For each $\ell,x$, it is convenient to define $T_m(\ell,x) = \inf \{ n : \Gamma_n(\ell,x) = m \}$.  By the standard strong law of large numbers, it follows that $\hat{P}_{T_m(\ell,x)}(\ell,x,y)  \to P(\ell,x,y)$ almost surely; 
this is because $\hat{P}_{T_m(\ell,x)}(\ell,x,y)$ is the sample average of $m$ independent Bernoulli random variables, each with success probability $P(\ell,x,y)$.  Now observe that for $n$ such that $T_m(\ell,x) \leq n < T_{m+1}(\ell,x)$, $\hat{P}_n(\ell,x,y) = \hat{P}_{T_m(\ell,x)}(\ell,x,y)$; i.e., between successive visits to state $x$ in which policy $\ell$ is sampled, $\hat{P}_n(\ell,x,y)$ remains constant.  It follows therefore that $\hat{P}_n(\ell,x,y) \to P(\ell,x,y)$ almost surely as well.

We now use a compactness argument analogous to that used to establish \eqref{eq:pi_convergence} to show that $\hat{\pi}_n(\ell) \to \pi(\ell)$ almost surely.  Let $J$ be the first $n$ at which $\hat{P}_n(\ell)$ is irreducible for both $\ell = 1,2$.  The time $J$ is almost surely finite, because both chains are sampled with equal probability until time $J$, and because $P(\ell)$ is irreducible for $\ell = 1,2$.  Thus for the remainder of our argument, we condition on the almost sure event $J < \infty$.  Next, consider any subsequence $\{n_k\}$ along which, almost surely, $\hat{\pi}_{n_k}(\ell) \to \pi'(\ell)$.  (Note that in general, this is a random subsequence.)  Since $\hat{\pi}_{n_k}(\ell) \hat{P}_{n_k}(\ell) = \hat{\pi}_{n_k}(\ell)$ for all $k$, almost sure convergence of $\hat{P}_n(\ell)$ implies that $\pi'(\ell) P(\ell) = \pi'(\ell)$.  Thus $\pi'(\ell) = \pi(\ell)$ almost surely.  Since this is almost surely true for every convergent subsequence, we conclude that $\hat{\pi}_n(\ell) \to \pi(\ell)$ almost surely, as required.

Because rewards are bounded, and thus in particular have finite moments, an argument analogous to that above for $\hat{P}_n$ establishes that almost surely we have:
\[ \hat{r}_n(\ell,x) \to r(\ell,x) \]
and 
\[ \hat{t}_n(\ell,x,y) - \hat{s}_n^2(\ell,x,y)^2 \to \Var(R_1 | A_0 = \ell, X_0 = x, X_1 = y). \]

When $J < \infty$, since each $\hat{P}_n(\ell)$ is irreducible, it follows that $\big(I-\hat{P}_n(\ell) + \hat{\Pi}_n(\ell)\big)^{-1}$ exists.  By continuity, conditioning on $J < \infty$, we have:
\[ \hat{\tilde{g}}(\ell,x) \to \tilde{g}(\ell,x) \]
almost surely as well, and thus:
\[ \hat{\sigma}^2(\ell,x) \to \sigma^2(\ell,x) \]
almost surely.

We now establish almost sure convergence of $\hat{\kappa}_n$ to $\kappa^*$.  To do this, for a distribution $\tilde{\pi}$ on the state space $S$ and a nonnegative vector $\tilde{\sigma}$, define the correspondence $K(\tilde{\pi},\tilde{\sigma})$ to be the set of minimizers of $\sum_{\ell = 1,2} \sum_{x \in S} \tilde{\pi}^2(\ell,x)\tilde{\sigma}^2(\ell,x)/\kappa(\ell,x)$ over $\kappa \in \mathcal{K}$; recall that $\mathcal{K}$ is compact so this correspondence is nonempty everywhere.  Further, observe that if $\tilde{\pi}$ and $\tilde{\sigma}$ are positive in all coordinates, then the minimizer is unique, i.e., $K$ is a function.  
Then by Lemma \ref{lem:maximizer} below, $K$ is continuous in $\tilde{\pi}$ and $\tilde{\sigma}$ when they are both positive in all coordinates.  Since $\hat{\pi}_n(\ell) \to \pi(\ell)$ and $\hat{\sigma}_n^2(\ell,x) \to \sigma^2(\ell,x)$ almost surely, and both limits are positive in all coordinates, it follows that $K(\hat{\pi}_n, \hat{\sigma}_n) \to K(\pi, \sigma) = \kappa^*$ almost surely, and thus $\hat{\kappa}_n \to \kappa^*$ almost surely.

In particular, we thus know that almost surely, $\hat{\kappa}_n(\ell,x) > 0$ for all sufficiently large $n$.  As a result, it follows that $\hat{p}_n(\ell,x) \to p^*(\ell,x)$ almost surely, where:
\[ p^*(\ell,x) = \frac{\kappa^*(\ell,x)}{\kappa^*(1,x) + \kappa^*(2,x)}. \]

To complete the proof, we require some additional notation.  We define the following stochastic matrix:
\[ Q(x,y) = p^*(1,x) P(1,x,y) + p^*(2,x) P(2,x,y). \]
Note that this matrix is irreducible, and because $\kappa^* \in \mathcal{K}$, we can easily see that $Q$ has the unique stationary distribution given by:
\[ \zeta^*(x) = \kappa^*(1,x) + \kappa^*(2,x). \]
(See also the discussion in Remark \ref{rem:markov_TAR}.)

In addition, we define:
\[ \hat{Q}_n(x,y) = \frac{\sum_{j = 1}^n I(X_{j-1} = x, X_j = y)}{\max\{M_n(x),1\}}. \]
Observe that $\hat{Q}_n$ is a stochastic matrix.  

We now show that $\hat{Q}_n \xrightarrow{p} Q$.  We rewrite $\hat{Q}_n(x,y)$ as follows:
\begin{equation}
\label{eq:hatQrewrite}
 \hat{Q}_n(x,y) = \sum_{\ell = 1,2} \hat{P}_{n}(\ell,x,y) \cdot \frac{\sum_{j = 1}^n I(X_{j-1} = x, A_{j-1} = \ell)}{\max\{M_n(x),1\}}.
\end{equation}

For each $x$ and $m$, let $S_m(x) = \inf \{ n \geq 0 : M_n(x) = m \}$; this is the time step at which the $m$'th visit to $x$ takes place.  Further, define $\tilde{A}_m = A_{S_m(x)}$; this is the policy sampled at the $m$'th visit to $x$.  Let $\mathcal{H}_m(x) = \sigma((X_j, U_j, V_j, j < S_m(x); X_{S_m(x)}))$ be the sigma field generated by randomness up to the $m$'th visit to $x$, but prior to the policy being chosen.  Finally, let $\hat{q}_m(\ell,x) = \hat{p}_{S_m(x)}(\ell,x)$.  Now observe that when $M_n(x) = m \geq 1$, we have:
\begin{align*}
 \frac{\sum_{j = 1}^n I(X_{j-1} = x, A_{j-1} = \ell)}{\max\{M_n(x),1\}} &= \frac{\sum_{i = 1}^m I(\tilde{A}_i = \ell)}{m}\\
 &=  \frac{\sum_{i = 1}^m I(\tilde{A}_i = \ell) - \hat{q}_i(\ell,x)}{m} + \frac{\sum_{i = 1}^m \hat{q}_i(\ell,x)}{m}.
\end{align*}
 The terms in the first sum on the right hand side of the previous expression form a martingale difference sequence adapted to $\mathcal{H}_i$.  Thus using orthogonality of martingale differences, we have:
 \[ \frac{1}{m^2}E\left \{ \left( \sum_{i = 1}^m I(\tilde{A}_i = \ell) - \hat{q}_i(\ell,x)\right)^2 \right \} \leq \frac{1}{4m}, \]
 which approaches zero as $m \to \infty$.  By Chebyshev's inequality, it follows that:
 \[ \frac{\sum_{i = 1}^m I(\tilde{A}_i = \ell) - \hat{q}_i(\ell,x)}{m} \xrightarrow{p} 0 \]
 as $m \to \infty$.  On the other hand, note that since $M_n(x) \to \infty$ almost surely, we also know that $S_m(x) \to \infty$ as $m \to \infty$ almost surely.  Thus it follows that:
 \[ \frac{\sum_{i = 1}^m \hat{q}_i(\ell,x)}{m} \to p^*(\ell,x) \]
almost surely as $m \to \infty$, and thus in probability as well.  Combining these insights, we conclude that:
\[  \frac{\sum_{j = 1}^n I(X_{j-1} = x, A_{j-1} = \ell)}{\max\{M_n(x),1\}} \xrightarrow{p} p^*(\ell,x) \]
as $n \to \infty$, and so returning to \eqref{eq:hatQrewrite}, we find that:
\[ \hat{Q}_n(x,y) \xrightarrow{p} \sum_{\ell=1,2} p^*(\ell,x) P(\ell,x,y)  = Q(x,y). \]

Next, observe that:
\begin{align*}
\frac{M_n(x)}{n} &= \frac{\sum_{j = 1}^n I(X_j = x)}{n} + \frac{I(X_0 = x) - I(X_n = x)}{n} \\
&= \left( \sum_{y \in S} \hat{Q}_n(x,y) \cdot \frac{\max\{M_n(y),1\}}{n} \right) + O_p\left(\frac{1}{n}\right).
\end{align*}

Since $M_n(x) \to \infty$ almost surely, in what follows we condition on $M_n(x) \geq 1$ for all $x$ and thus ignore the ``$\max$'' on the right hand side in the preceding expression.  Note that for all $n$, $\sum_{x \in S} M_n(x) = n$.  Thus using a compactness argument analogous to that used to establish \eqref{eq:pi_convergence}, it follows that:
\[ \frac{M_n(n)}{n} \xrightarrow{p} \zeta^*(x). \]

We can now complete the proof of the theorem.  We have:
\begin{align}
    \frac{1}{n} \Gamma_n(\ell,x)=& \frac{1}{n} \sum_{j=0}^{n-1} I(X_j=x, A_j=\ell) \nonumber\\
    =& \frac{1}{n} \sum_{j=0}^{n-1} I(X_j=x) p^*(\ell,x)+\frac{1}{n} \sum_{j=0}^{n-1} I(X_j=x) \big(\hat{p}_j(\ell,x)-p^*(\ell,x)\big)\nonumber\\
    &+\frac{1}{n} \sum_{j=0}^{n-1} I(X_j=x)\big(I(A_j=\ell)-\hat{p}_j(\ell,x)\big)\label{eq: Gamma for algorithm conv}
\end{align}
Because $I(X_j=x)\big(I(A_j=\ell)-\hat{p}_j(\ell, x)\big)$ is a martingale difference measurable with respect to $\mathcal{G}_j$, orthogonality of martingale differences implies that 
\begin{align}
    E\bigg\{\bigg(&\frac{1}{n}\sum_{j=1}^{n} I(X_j=x)\big(I(A_j=\ell)-\hat{p}_j(\ell,x)\big)\bigg)^2\bigg\}\\
    &\leq  E\bigg\{\frac{1}{4} \cdot \frac{1}{n^2} \Gamma_{n}(\ell,x) \bigg\} \leq \frac{1}{4n}\\
    &\rightarrow  0
\end{align}
as $n\rightarrow \infty$. Therefore, by Chebyshev's inequality 
\begin{align}\label{eq: convergence P(l) martingale 2}
    \frac{1}{n} \sum_{j=0}^{n-1} I(X_j=x)\big(I(A_j=\ell)-\hat{p}_j(
\ell,x)\big)\xrightarrow{p} 0
\end{align}
as $n\rightarrow \infty$. Also, since $\hat{p}_n(\ell,x) \to p^*(\ell,x)$ almost surely, we have: 
\begin{align}\label{eq: constant c 1 x}
    \frac{1}{n} \sum_{j=0}^{n-1} I(X_j=x) \big(\hat{p}_j(\ell,x)-p^*(\ell,x)\big)\xrightarrow{p} 0.
\end{align}
Finally,
\[ \frac{1}{n} \sum_{j=0}^{n-1} I(X_j=x) p^*(\ell,x) = \frac{p^*(\ell,x) M_n(\ell,x)}{n} \xrightarrow{p} p^*(\ell,x) \zeta^*(\ell,x). \]
Combining the preceding results, we conclude that as $n \to \infty$ in (\ref{eq: Gamma for algorithm conv}), we have
\begin{equation}\label{eq: Gamma n with constant c}
    \frac{1}{n} \Gamma_n(\ell,x)\xrightarrow{p} \zeta^*(\ell,x)p^*(\ell,x) = \kappa^*(\ell,x)
\end{equation}
as $n\rightarrow \infty$, completing the proof of the theorem.
\end{proof}

\begin{lemma}
\label{lem:maximizer}
Suppose that the set $X$ is compact, the set $\Theta$ is open, and the real-valued function $f(\theta,x)$ is  continuous on the domain $\Theta \times X$.  Suppose further that for every $\theta \in \Theta$, there exists a unique $x^*(\theta) = \arg \min_{x \in X} f(\theta,x)$.  Then $x^*(\theta)$ is continuous in $\theta$.
\end{lemma}

\begin{proof}
Suppose that $\theta^{(n)} \to \theta$.  For all $n$ we have:
\begin{equation}
    \label{eq:f_lowerbound}
    f(\theta^{(n)}, x^*(\theta^{(n)})) \leq f(\theta^{(n)}, x^*(\theta)).
\end{equation}
Since $X$ is compact, let $\{n_k\}$ be a subsequence such that $x^*(\theta^{(n_k)}) \to x'$ as $k \to \infty$.   Taking limits on both sides of \eqref{eq:f_lowerbound} along the sequence $\{n_k\}$, we obtain:
\[ f(\theta, x') \leq f(\theta, x^*(\theta)). \]
Since $x^*(\theta)$ is unique, this is only possible if $x' = x^*(\theta)$.  Since every convergent subsequence must have the limit $x'$, we conclude that $x^*(\theta^{(n)}) \to x^*(\theta)$ as $n \to \infty$, as required.
\end{proof}

\end{document}